\documentclass[11pt]{article}

\usepackage{lipsum}

\usepackage[english]{babel}
\usepackage{amsthm,amsmath,amsfonts,amssymb}
\usepackage{amsmath,amssymb,amsfonts}
\usepackage{latexsym,color}
\usepackage{float}
\usepackage[margin=1in]{geometry}
\floatstyle{ruled}
\usepackage{authblk}
\usepackage{colonequals}
\usepackage[normalem]{ulem}

\usepackage{hyperref}
\hypersetup{
    bookmarksnumbered=true, 
    unicode=false, 
    pdfstartview={FitH}, 
    pdftitle={Title}, 
    pdfauthor={Authors}, 
    pdfsubject={}, 
    pdfcreator={}, 
    pdfproducer={}, 
    pdfkeywords={}, 
    pdfnewwindow=true, 
    colorlinks=true, 
    linkcolor=blue, 
    citecolor=blue, 
    filecolor=blue, 
    urlcolor=blue 
}

\newtheorem{theorem}{Theorem}[section]
\newtheorem{lemma}[theorem]{Lemma}

\newtheorem{corollary}[theorem]{Corollary}

\newcommand{\eq}[1]{\hyperref[eq:#1]{(\ref*{eq:#1})}}
\renewcommand{\sec}[1]{\hyperref[sec:#1]{Section~\ref*{sec:#1}}}
\newcommand{\thm}[1]{\hyperref[thm:#1]{Theorem~\ref*{thm:#1}}}
\newcommand{\lem}[1]{\hyperref[lem:#1]{Lemma~\ref*{lem:#1}}}
\newcommand{\cor}[1]{\hyperref[cor:#1]{Corollary~\ref*{cor:#1}}}
\newcommand{\app}[1]{\hyperref[app:#1]{Appendix~\ref*{app:#1}}}
\newcommand{\tab}[1]{\hyperref[tab:#1]{Table~\ref*{tab:#1}}}

\def\ket#1{{\lvert}#1\rangle}

\renewcommand{\>}{\rangle}

\renewcommand{\(}{\left(}
\renewcommand{\)}{\right)}

\def\abs#1{\left| #1 \right|}
\def\span{\mbox{span}}
\newcommand{\eps}{\varepsilon}
\def\th#1{$#1^\mathrm{th}$}
\renewcommand{\th}[1]{${#1}^{\textrm{th}}$}

\begin{document}

\title{Nested Quantum Walks with Quantum Data Structures%
\thanks{Partially supported by the French ANR Defis project
ANR-08-EMER-012 (QRAC),
the European Commission IST STREP project
25596 (QCS),
NSERC,
NSERC Frequency,
MITACS,
and
the Ontario Ministry of Research and Innovation.}}

\author[1,2]{Stacey Jeffery\thanks{\texttt{sjeffery@uwaterloo.ca}}}
\author[1,2]{Robin Kothari\thanks{\texttt{rkothari@cs.uwaterloo.ca}}}
\author[3]{Fr\'ed\'eric Magniez\thanks{\texttt{frederic.magniez@univ-paris-diderot.fr}}}
\affil[1]{\small David R.\ Cheriton School of Computer Science, University of Waterloo, Canada}
\affil[2]{Institute for Quantum Computing, University of Waterloo, Canada}
\affil[3]{\small CNRS, LIAFA, Univ Paris Diderot, Sorbonne Paris-Cit\'e, 75205 Paris, France}

\date{}

\maketitle

\begin{abstract}
We develop a new framework that extends the quantum walk framework of Magniez, Nayak, Roland, and Santha, by utilizing the idea of quantum data structures to construct an efficient method of nesting quantum walks. Surprisingly, only classical data structures were considered before for searching via quantum walks. 

The recently proposed learning graph framework of Belovs has yielded improved upper bounds for several problems, including triangle finding and  more general subgraph detection. We exhibit the power of our framework by giving a simple explicit constructions that reproduce both the $O(n^{35/27})$ and $O(n^{9/7})$  learning graph upper bounds (up to logarithmic factors) for triangle finding, and discuss how other known upper bounds in the original learning graph framework can be converted to algorithms in our framework. We hope that the ease of use of this framework will lead to the discovery of new upper bounds.

\end{abstract}

\section{Introduction}
\subsection{Background}
Shenvi, Kempe, and Whaley~\cite{ShenviKW03} first
pointed out the algorithmic potential of quantum walks by designing a walk based
algorithm to emulate Grover search~\cite{gro96}. 
The first algorithm using quantum walks that goes beyond the capability of
Grover search is due to Ambainis~\cite{amb04} for element distinctness.  
In his seminal paper he resolved the quantum query complexity of the problem, settling a difficult question that had been open for several
years~\cite{bdh05,AS04}.
Based on Ambainis' quantum walk, Szegedy~\cite{sze04} developed a general theory of quantum walk based algorithms.
He constructed the quantum analogue $W(P)$ of any symmetric, ergodic Markov chain $P$,
and designed a quantum search algorithm based on the use of $W(P)$.
Later, this was generalized to reversible, ergodic Markov chains $P$,
and the construction simplified by the use of quantum phase estimation~\cite{MNRS11,mnrs12}.

After a series of new quantum walk algorithms with substantially better complexity in a variety of contexts \cite{BS06,MN07,MSS07,ACRSZ10,CK11,leg12},
a new framework was introduced by Belovs, known as learning graphs~\cite{bel11}.
Learning graphs are a model of computation related to span programs~\cite{rei09}.
Introduced by Reichardt in the model of quantum query algorithms, span programs have been used to characterize
quantum query complexity in terms of the negative adversary method~\cite{rei09,rei11,LMR+11}.
Whereas explicitly constructing a span program is difficult in practice, coming up with a learning graph is more easily accomplished.
Indeed, Belovs~\cite{bel11} gave a learning graph whose query complexity improved the best known quantum query algorithms
for triangle finding. 
Later on, other learning graphs were constructed for 
constant-size subgraph finding~\cite{zhu11,LMS11},
 and also for $k$-element distinctness~\cite{bl11,bel12}.
Recently, a simpler and more efficient learning graph was proposed for triangle finding,
together with a new application to associativity testing~\cite{LMS12}.

In order to get an explicit quantum algorithm, one has to transform a learning graph into a span program, and then to an algorithm.
This transformation is generic but incorporates many technical difficulties. The resulting algorithm is quite complicated. 
For instance, the time and space complexities may not be optimal. For specific problems, such as $st$-connectivity~\cite{br12}, 
one can find explicit implementations of span programs.
Still, one would like to have a framework which leads easily to the explicit algorithm and its complexity 
(not only query but also time and space, for instance), as is the case in the setting of search algorithms based on quantum walks.

\subsection{Quantum walk based search algorithms}
At an abstract level, many search problems may be cast as the problem of finding a ``marked'' element from a set~$\Omega$ with~$n$ elements. 
Let~$M\subseteq \Omega$ be the set of marked elements. 
Assume that the search algorithms maintain a data structure~$\mathcal{D}$ that associates some
data~$\mathcal{D}(u)$ with every state $u \in \Omega$.  {From} $\mathcal{D}(u)$, we would like to determine if $u \in M$.  
In this framework, Grover search, and more generally amplitude amplification, consists of the following procedure,
where $\pi$ is some given probability distribution on $\Omega$.
\begin{enumerate}
\item Set up a superposition $\ket{\pi}=\sum_u\sqrt{\pi_u}\ket{u}$ of elements $u\in \Omega$\\
Compute $\mathcal{D}$, that is $\ket{u}\mapsto \ket{u}\ket{\mathcal{D}(u)}$, for all elements $u$ in the current superposition
\item Iterate ${\frac{1}{\sqrt\eps}}$ times, where $\eps=\Pr_\pi (u\in M)$.
\begin{enumerate}
\item Flip marked elements, that is $\ket{u}\ket{\mathcal{D}(u)}\mapsto -\ket{u}\ket{\mathcal{D}(u)}$, if $u\in M$
\item Uncompute $\mathcal{D}$, that is $\ket{u}\ket{\mathcal{D}(u)}\mapsto \ket{u}$\\
Apply a reflection through $\ket{\pi}$\\
Compute $\mathcal{D}$, that is $\ket{u}\mapsto \ket{u}\ket{\mathcal{D}(u)}$, for all elements $u$ in the current superposition
\end{enumerate}
\item Output $u$ if it is marked, otherwise output `no marked element'
\end{enumerate}

When operating with $\mathcal{D}$,  we distinguish two types of cost.
\begin{itemize} 
\item[] \textbf{Setup cost $\mathsf{S}$:}
The cost of creating $\ket{\pi}$ and of constructing the data structure~$\mathcal{D}(u)$ for the state~$u$.
\item[] \textbf{Checking cost $\mathsf{C}$:} The cost of checking
  if~$u \in M$ using~$\mathcal{D}(u)$.
\end{itemize}
These costs may be thought of as vectors listing all the measures of complexity of interest, such as query and time complexity. 
Then the above algorithm has cost of order ${\frac{1}{\sqrt\eps}}(\mathsf{S}+\mathsf{C})$.
This above approach is a priori only valid when the probability $\Pr_\pi (u\in M)$ is known in
advance.  This assumption may be removed using standard techniques, without increasing the asymptotic complexity of the
algorithms~\cite{BBHT98}. Indeed, if only a lower bound $\eps > 0$ on $\Pr_\pi (u\in M)$ is known for non-empty~$M$, then the above argument  
can be modified in order to determine if $M$ is empty or find an element of $M$.

A more cost-effective approach is to implement the reflection through $\ket{\pi}$ using the quantum analogue $W(P)$
of a reversible and ergodic walk $P$ on a given graph $G$. 
In this way, we can ensure that $u$ has only small changes at each step, permitting us to update its data structure $\mathcal{D}(u)$ without uncomputing/computing it.
We then introduce a third type of cost.
\begin{itemize}
\item[] \textbf{Update cost $\mathsf{U}$:}
The cost of simulating one application of $W(P)$ 
according to the Markov chain~$P$,
and of updating $\mathcal{D}(u)$ to~$\mathcal{D}(v)$ where $(u,v)\in G$.
\end{itemize}
Then finding a marked element, if there is one, can be done with cost of order
$\mathsf{S}+{\frac{1}{\sqrt\eps}}({\frac{1}{\sqrt\delta}}\mathsf{U}+\mathsf{C})$,
where $\delta$ is the spectral gap of $P$ \cite{MNRS11}.

\subsection{Our contribution}
In this paper, we introduce a new model of quantum walks which can reproduce the best known learning graph upper bounds for triangle finding and related problems (up to log factors). 
Our framework is also simple and leads to explicit algorithms that are much simpler to analyze than learning graphs.

To achieve our goal, we introduce the notion of quantum data structures for nested quantum walks.
Surprisingly, only classical data structures were considered before for searching via quantum walk.
We demonstrate the potential of this framework by giving 
 explicit and simple quantum algorithms for triangle finding, achieving the same query complexity as their learning graph analogues up to polylogarithmic factors (\sec{extensions}) and more generally, reproducing all subgraph containment upper bounds of \cite{LMS12} up to polylogarithmic factors, and potentially improving on these upper bounds for certain graphs (\sec{subgraph}).

Concretely, assume that we introduce another quantum walk based search as a checking subroutine.
Let $P_1$ and $P_2$ be, respectively, the outer and inner walks, with respective costs $\mathsf{S}_1,\mathsf{U}_1,\mathsf{C}_1$
and $\mathsf{S}_2,\mathsf{U}_2,\mathsf{C}_2$, and parameters $\eps_1,\delta_1$ and $\eps_2,\delta_2$.
Since the checking cost of the outer walk can be expressed as 
$\mathsf{C}_1=O(\mathsf{S_2}+{\frac{1}{\sqrt\eps_2}}({\frac{1}{\sqrt\delta_2}}\mathsf{U_2}+\mathsf{C_2}))$
the global of cost of finding a marked element is of order (neglecting log factors)
$$ \mathsf{S}_1+{\frac{1}{\sqrt\eps_1}}\left({\frac{1}{\sqrt\delta_1}}\mathsf{U_1}+\mathsf{S_2}+{\frac{1}{\sqrt\eps_2}}\left({\frac{1}{\sqrt\delta_2}}\mathsf{U_2}+\mathsf{C_2}\right)\right).$$
We will show that allowing a quantum data structure for the outer walk, one can push the setup cost $\mathsf{S}_2$ of the inner walk to the outer one. In order to accomplish this, we include in the data structure $\mathcal{D}(u)$, for each state $u$ of the outer walk, the initial state for the inner quantum walk. It is then necessary that the initial state for the inner walk be properly preserved by the outer update, as well as the inner walk, so that we needn't perform the inner setup multiple times. In equations, this leads to an improved cost of order (\thm{twolevel}):
$$ \mathsf{S}_1+\mathsf{S_2}+{\frac{1}{\sqrt\eps_1}}\left({\frac{1}{\sqrt\delta_1}}\mathsf{U_{1}'}+{\frac{1}{\sqrt\eps_2}}\left({\frac{1}{\sqrt\delta_2}}\mathsf{U_2}+\mathsf{C_2}\right)\right),$$
 where some logarithmic factors are now omitted. The new update cost $\mathsf{U}_{1}'$ is the cost for updating the two data structures
 when one step of the outer walk is performed. In our application, $\mathsf{U}_{1}'$ will be of the same order as $\mathsf{U}_1$. Note that the update of the inner walk may be performed much more often than the update of the outer walk. Some of the power of our nested framework comes from keeping expensive operations in the outer update, and pushing only cheaper operations into the inner update.

Our very simple idea seems for now as powerful as most learning graph based 
algorithms~\cite{bel11,bl11,zhu11,LMS11,LMS12}, 
except the recent one for $k$-element distinctness~\cite{bel12}, which deviates significantly from the original learning graph framework. 
Nonetheless, we conjecture that one might use our framework, possibly with further extensions, to design a similar algorithm.

The idea is a posteriori very natural. Let us make an analogy with the adversary method. The method was first designed with positive weights
by Ambainis~\cite{amb00}. Only several years later, the possibility to put negative weight was proven to be not only possible but useful. 
Later, using the dual of this method, span programs were proven to 
achieve optimal quantum query complexity.
In a more modest direction, we hope that our contribution will lead to a better understanding of the power of quantum walks and learning graphs.

While quantum walks are a powerful computational primitive, our main motivation for extending the quantum walk framework to include recent learning graph algorithms was not to demonstrate that fact. We believe that the framework will be useful for designing new quantum algorithms, since it allows the use of intuition  gained from experience with quantum and random walks. Some ideas, like function composition and subroutines, are very natural in the quantum walk framework, but not as straightforward in the learning graph framework. We hope that the different perspectives provided by the quantum walk framework and the learning graph framework will lead to the discovery of new quantum algorithms.

\section{Preliminaries}
\label{sec:prelim}
\subsection{Model and notation}
Our framework  applies to other notions of complexity such as time and space. Nonetheless we will only illustrate it within the model of query complexity.
In the model of quantum query complexity, the input string $x$ is accessed by querying a black box to learn the bits of $x$. More concretely, the black box performs the following unitary operation: $\mathcal{O}_x: |j,b\> \mapsto |j,b\oplus x_j\>$. The quantum query complexity of a quantum algorithm is the number of times the algorithm uses the operator  $\mathcal{O}_x$. In this paper, $x$ will always refer to the input string. 
For a more thorough introduction to the quantum query model, see \cite{BBCMdW01}.

For problems where the input is a graph, such as the triangle finding problem, $x$ will be the adjacency matrix of an undirected graph. The triangle finding problem asks if the input graph contains a \emph{triangle}, i.e., three pairwise adjacent vertices. A vertex is called a \emph{triangle vertex} if it is part of a triangle. Similarly, an edge is called a \emph{triangle edge} if it is part of a triangle. The best known triangle finding algorithms solve, as a subroutine, the graph collision problem. In the graph collision problem parametrized by a graph $G$ on $n$ vertices, the input is a binary string of length $n$ that indicates whether a vertex in $G$ is marked or not. The objective is to determine whether there exist two adjacent vertices in $G$ that are both marked. Note that the graph $G$ is part of the problem specification, and not part of the input.

We will use the $\tilde{O}$ notation to suppress logarithmic factors. More precisely, $f(n) = \tilde{O}(g(n))$ means $f(n) = O(g(n) \log^k n)$ for some integer $k$.

Given any algorithm computing a function $f(x)$ with probability at least $2/3$, we can repeat this algorithm $O\(\log\(\frac{1}{\epsilon}\)\)$ times and take the majority vote of outcomes to boost the algorithm's success probability to $1-\epsilon$. Since all the subroutines in this paper are used at most $\textrm{poly}(n)$ times, if we boost a bounded error subroutine's success probability to $1-1/\textrm{poly}(n)$, then the final algorithm will still be correct with bounded error. This adds one log factor to the complexity of the subroutine; this is the source of the log factors in this paper. We will also use this technique to boost the success probability of quantum algorithms that do not output a bit, but instead flip the phase when $f(x)$ is 1. By standard techniques, such an algorithm can be converted to one that computes $f(x)$ in a separate register. Now we can repeat this algorithm logarithmically many times, flip the phase conditioned on the majority vote of outcomes and uncompute the answers to obtain an algorithm whose success probability has been boosted to $1-1/\textrm{poly}(n)$.

The Johnson graph $J(n,r)$ is a graph on $\binom{n}{r}$ vertices, where each vertex is an $r$-subset of $\{1,2,\ldots,n\}$, and two vertices are adjacent if their symmetric difference contains exactly two elements. The spectral gap of $J(n,r)$ is $\Omega(1/r)$.

Let $G=(V,E)$ be a graph, and let $R,S\subseteq V$ and $F$ a set of edges on $V$, not necessarily contained in $E$. We denote by $G_{R}$  the subgraph of $G$ induced by $R$, and by $G_{R,S}$ the bipartite subgraph of $G$ induced by classes $R$ and $S$. That is, $G_{R,S}$ has vertex set $R\cup S$ and edges set $R\times S\cap E$. Denote by $G(F)$ the subgraph of $G$ obtained by restricting to edges in $F$. That is, $G(F)$ has vertex set $V$ and edges set $E\cap F$.

\subsection{Review of the quantum walk framework}

We now review the quantum walk search framework of Magniez, Nayak, Roland and Santha~\cite{MNRS11}. We present a slight modification of the framework, which is tailored to our application. We discuss the minor differences between our presentation and the original framework at the end of this section.

Let $P$ be a reversible Markov chain on (finite) state space $\Omega$, with stationary distribution $\pi$. Let $\mathcal{D}_x:\Omega\rightarrow \mathcal{H}_D$, for some Hilbert space $\mathcal{H}_D$, be a mapping that associates data with a state of the Markov chain for input $x$ (i.e., for $u \in \Omega$, $|\mathcal{D}_x(u)\>$ is the data associated with state $u$ for input $x$). We will refer to $\{|\mathcal{D}_x(u)\>\}_{u \in \Omega}$ as the data structure associated with the Markov chain.

In all previous applications of the quantum walk search framework, the data structure has always been classical, i.e., $\ket{\mathcal{D}_x(u)}$ is a classical basis state for all $u \in \Omega$, however, this needn't be the case in general. In our framework, we make use of quantum data structures.

\paragraph{Setup} Define $\mathcal{H}_L\cong\mathcal{H}_R\cong\mathcal{H}_\Omega\colonequals \span\{\{\ket{0}\}\cup\{\ket{u}:u\in\Omega\}\}$.
Define the initial state $\ket{\pi^x}$ as $\ket{\pi^x}\colonequals\sum_{u\in\Omega}\sqrt{\pi_u}\ket{u}_L\ket{0}_R\ket{\mathcal{D}_x(u)}_D$
and let $\mathsf{S}=\mathsf{S}(P,\mathcal{D}_x)$ be the cost of constructing $\ket{\pi^x}$. We call $\mathsf{S}$ the \emph{setup cost}.

\paragraph{Update}
Define two unitary operators $\mathcal{U}_P$ and $\mathcal{U}_D$ with the following action for all $u\in\Omega$:
$$\mathcal{U}_P:\ket{u}_L\ket{0}_R\mapsto \ket{u}_L\sum_{v\in\Omega}\sqrt{P_{uv}}\ket{v}_R,
\quad\text{and}\quad 
\mathcal{U}_D:\ket{u}_L\ket{v}_R\ket{\mathcal{D}_x(u)}\mapsto \ket{u}_L\ket{v}_R\ket{\mathcal{D}_x(v)}.$$
In \cite{MNRS11}, these two operations are defined as one, however, we find it beneficial to consider them separately, since they play different roles: $\mathcal{U}_P$ updates the walk according to the Markov chain $P$, while $\mathcal{U}_D$ updates the data structure. Since the Markov chain is independent of the input in all of our applications, the query complexity of implementing $\mathcal{U}_P$ is 0, although its time complexity is non-zero.

Let $\mathsf{U}=\mathsf{U}(P,\mathcal{D}_x)$ denote the cost of implementing $\mathcal{U}_D\mathcal{U}_P$. We refer to $\mathsf{U}$ as the \emph{update cost}.

\paragraph{Checking}
Fix some set $M_x\subseteq \Omega$ of marked elements. Let $\mathsf{C}=\mathsf{C}(M_x)$ be the cost of implementing the unitary that acts as
$$\ket{u}_\Omega\ket{\mathcal{D}_x(u)}_D\mapsto\left\{\begin{array}{ll}
-\ket{u}_\Omega\ket{\mathcal{D}_x(u)}_D & \mbox{if }u\in M_x\\
\ket{u}_\Omega\ket{\mathcal{D}_x(u)}_D & \mbox{else}\end{array}\right..$$
We refer to $\mathsf{C}$ as the \emph{checking cost}.
With these definitions, we are ready to state the main theorem of \cite{MNRS11}.

\begin{theorem}
\label{thm:mnrs1}
Let $\delta>0$ be the eigenvalue gap of $P$. Let $\eps$ be such that for any $x$, if $M_x$ is non-empty, then $\pi(M_x)\geq\eps$. Then there is a quantum algorithm with cost 
$$O\(\frac{1}{\sqrt{\eps}}\left(\frac{1}{\sqrt{\delta}}\mathsf{U}+\mathsf{C}\right)\)$$
that, when given the initial state $\ket{\pi^x}$, determines with bounded error whether $M_x$ is non-empty. More precisely, the algorithm implements (with bounded error) the map
$$\ket{\pi^x}\mapsto\left\{\begin{array}{ll}
-\ket{\pi^x} & \mbox{if } M_x\neq\emptyset\\
\ket{\pi^x} & \mbox{else}\end{array}\right. .$$
\end{theorem}

We will use of the above theorem, which is proven, though not explicitly stated in these terms, in \cite{MNRS11}. The explicit theorem of \cite{MNRS11} follows directly from this theorem, and is stated in \cor{mnrs}.

\begin{corollary}
\label{cor:mnrs}
There is a quantum algorithm with cost
$$O\(\mathsf{S}+\frac{1}{\sqrt{\eps}}\left(\frac{1}{\sqrt{\delta}}\mathsf{U}+\mathsf{C}\right)\)$$
that determines with bounded error whether $M_x$ is non-empty.
\end{corollary}

This presentation differs from the original presentation in some aspects. First, we allow the data structure to be quantum, while the original presentation implicitly assumed it to be classical. This, however, does not affect the proof. Second, the original presentation had a data structure associated with both registers of the walk, whereas we have only one data structure. As mentioned in~\cite{MNRS11}, this is only for convenience. Third, we explicitly state the cost of the algorithm as a sum of two costs, the setup cost and the cost of running the walk. This is consistent with the original algorithm.

\section{Two-level nested quantum walks}
\label{sec:framework}

In this section we show how to augment the above framework with nested quantum walks using a quantum data structure. As motivation for this framework, 
it will be useful to compare the quantum walk algorithm of Magniez, Santha and Szegedy for triangle finding~\cite{MSS07}, which has query complexity $\tilde{O}(n^{1.3})$, with the learning graph algorithm of Belovs~\cite{bel11}, which has query complexity $O(n^{35/27}) = O(n^{1.297})$.
In \sec{twolevel}, we formally state the general construction for the case of two-level nested quantum walks. In \sec{gen}, we explain how to generalize to $k$  levels of nesting.

\subsection{Warming up}\label{sec:warmingup}

\paragraph{Quantum walk algorithm for triangle finding}
In the quantum walk algorithm of \cite{MSS07}, a quantum walk is used to search for a pair of triangle vertices. The walk is on the Johnson graph of sets of $r$ vertices,  $J(n,r)$, which has spectral gap $\delta=\Omega\left(\frac{1}{r}\right)$. A state is marked if it contains two vertices of a triangle, so the proportion of marked states is lower bounded by $\eps=\Omega\left(\frac{r^2}{n^2}\right)$. 

If $G$ denotes the input graph, the data associated with a state $R$ of the Johnson graph is $\ket{\mathcal{D}_G(R)} = \ket{G_R}$, which encodes
the subgraph of $G$ induced by the vertices in $R$. The setup cost is $\mathsf{S}=O(r^2)$, since, for each state, we must initially query all $\binom{r}{2}$ possible edges in the subgraph induced by $R$. The update cost is $\mathsf{U} = O(r)$, since, at each step of the walk, we remove a vertex $i$ and add a new vertex $i'$ to the set of vertices, and must therefore test the adjacency of $i'$ with all remaining vertices in the set. 

Finally, to check whether we have two triangle vertices in a state $R$, we search for a vertex $k$ that forms a triangle with two of the vertices in $R$. To check if a fixed vertex $k$ forms a triangle with any two of the vertices in $R$, we solve the graph collision problem on $G_R$, which is encoded in $\mathcal{D}_G(R)$, with vertices in $R$ defined as being marked if they are adjacent to $k$. Thus, we can check if a vertex $k$ forms a triangle with two vertices in $R$ using $O(r^{2/3})$ queries, and so we can check if any vertex forms a triangle with two vertices in $R$ using $\mathsf{C}= O(\sqrt{n}r^{2/3})$ queries.

Plugging all of these values into the formula from \cor{mnrs} gives (neglecting logarithmic factors)
$$\begin{array}{rcl}\mathsf{S}+\frac{1}{\sqrt{\eps}}\left(\frac{1}{\sqrt{\delta}}\mathsf{U}+\mathsf{C}\right) &=& r^2+\frac{n}{r}\left(\sqrt{r}r+\sqrt{n}r^{2/3}\right)\\
&=&r^2+n\sqrt{r}+\frac{n^{3/2}}{r^{1/3}},\end{array}$$
which is optimized by setting $r=n^{3/5}$, yielding query complexity $\tilde{O}(n^{1.3})$ for triangle finding.

\paragraph{Learning graph algorithm for triangle finding}
In contrast, the recent learning graph algorithm of Belovs~\cite{bel11} for triangle finding considers sets of $r$ vertices, but only queries a sparsification of the possible edges between them. 
If the sparsification parameter is $s\in (0,1]$, then each potential edge is queried with probability $s$, giving an expected $s\binom{r}{2}$ queries for each state. 

We can try to apply this edge sparsification idea to the above quantum walk. We will redefine the data structure to be a random sparsification of the potential edges in $G_R$, so that the setup cost decreases to $\mathsf{S}=O(sr^2)$ and similarly the update cost decreases to $\mathsf{U}=O(sr)$ (let us assume that the setup and update operations can be performed with their average cost, instead of their worst-case cost).  However, now some states that contain two triangle vertices will not contain the edge between them, and will therefore not be detected by the checking procedure. Thus, the probability that a state is marked has also decreased to $\eps = \Omega(s\frac{r^2}{n^2})$. If we plug in these new values, we get query complexity 
$$\begin{array}{rcl}\mathsf{S}+\frac{1}{\sqrt{\eps}}\left(\frac{1}{\sqrt{\delta}}\mathsf{U}+\mathsf{C}\right)&=&sr^2+\frac{n}{r\sqrt{s}}\left(\sqrt{r}sr+\sqrt{n}r^{2/3}\right)\\
&=&sr^2+n\sqrt{rs}+\frac{n^{3/2}}{r^{1/3}\sqrt{s}},\end{array}$$
which is optimized by $s=1$ and $r=n^{3/5}$, yielding no improvement over the standard quantum walk algorithm. 

One key idea on which our framework is based, is to observe that a state $R$ that contains two triangle vertices should be considered marked, even if its data structure is missing the edge between them. Of course, we need that edge to know that the state is marked, but we can check this by searching for a \emph{good} data structure, i.e., one containing the triangle edge. To this end, we allow the data structure of a state $R$ to be a quantum superposition over a number of classical data structures, some of which may contain some kind of certificate for $R$, if $R$ is in fact marked. We walk on this data structure as a checking subroutine. Importantly, the setup of this inner walk, that is, the creation of an appropriate superposition of data structures, has already been taken care of as part of the setup of the outer walk. It is very important that we only perform this setup once, as opposed to at every checking step. 

In \sec{extensions}, we use our framework to effectively apply the sparsification idea to triangle finding in a nested quantum walk. 
We first present the general framework for nesting quantum walks with a quantum data structure. We start with nesting one walk within another. 

\subsection{Framework for nested quantum walks}
\label{sec:twolevel}

Let $P$ be a reversible Markov chain on (finite) state space $\Omega$, with stationary distribution $\pi$. For each input $x$, fix a marked set $M_x\subseteq\Omega$. Let $\mathcal{D}_x:\Omega\rightarrow \mathcal{H}_{D}$ be some associated data mapping. The Markov chain $P$ defines what we refer to as the \emph{outer walk}. 

Fix $u\in \Omega$. We will define what we refer to as the \emph{inner walk}, whose purpose is to check whether $u$ is in $M_x$. 
Let $P'$ be a reversible Markov chain on state space $\Omega'$ with stationary distribution $\pi'$.  Let $M_x^u\subseteq\Omega'$ be the marked set for the inner walk, with respect to $u$. Let $\mathcal{D}_x^u:\Omega'\rightarrow \mathcal{H}_{D'}$ be the inner data mapping with respect to $u$. Note that, although we have the walk itself, defined by $P'$, $\Omega'$ and $\pi'$, independent of $u$, the marked set and data mapping are dependent on $u$.

We suppose the following relationship between the inner and outer marked sets:
$$M_x=\{u\in\Omega:M_x^u\neq\emptyset\}.$$
We can then view the set $M_x^u$ as a set of witnesses for $u\in M_x$.

In order to make use of the quantum data structure, we suppose the following relationship holds for the inner and outer data structures. First, we require $\mathcal{H}_{D} = \mathcal{H}_{\Omega'}\otimes\mathcal{H}_{D'}$, and in particular, we need
$$\ket{\mathcal{D}_x(u)}=\sum_{s\in\Omega'}\sqrt{\pi'_s}\ket{s}\ket{\mathcal{D}_x^u(s)}.$$
That is, $\ket{\mathcal{D}_x(u)}$ is the initial state of the inner walk with respect to $u$. In general, the data structure for the outer walk could store more information than just the initial state of the inner walk, but this is all we need for our applications.

\paragraph{Setup} Define spaces $\mathcal{H}_{L}\cong\mathcal{H}_{R}\cong\mathcal{H}_{\Omega}$ and $\mathcal{H}_{L'}\cong\mathcal{H}_{R'}\cong\mathcal{H}_{\Omega'}$. The setup cost, $\mathsf{S}=\mathsf{S}(P,\mathcal{D}_x)$, is the cost of constructing the initial state:
\begin{eqnarray}
\ket{\pi^{x}} & = & \sum_{u\in\Omega}\sqrt{\pi_u}\ket{u}_{L}\ket{0}_{R}\ket{\mathcal{D}_x(u)}_{D}\ket{0}_{R'}\nonumber \\
& = & \sum_{u\in\Omega}\sqrt{\pi_u}\ket{u}_{L}\ket{0}_{R}\sum_{s\in\Omega'}\sqrt{\pi'_s}\ket{s}_{L'}\ket{\mathcal{D}_x^u(s)}_{D'}\ket{0}_{R'}.\nonumber
\end{eqnarray}
We emphasize that the initial state for the outer walk contains, in superposition, initial states for the inner walks with respect to each $u\in\Omega$. 

\paragraph{Update} The update cost, $\mathsf{U}=\mathsf{U}(P,\mathcal{D}_x)$, is the cost of implementing $\mathcal{U}_D\mathcal{U}_P$. Recall that $\mathcal{U}_P$ is defined by the action $\mathcal{U}_P:\ket{u}_{L}\ket{0}_{R}\mapsto \ket{u}_{L}\sum_{v\in\Omega}\sqrt{P_{u,v}}\ket{v}_{R}$. 
Similarly, $\mathcal{U}_D$ is defined by the action
$$\mathcal{U}_D:\ket{u}_{L}\ket{v}_{R}\ket{\mathcal{D}_x(u)}_{D}\mapsto \ket{u}_{L}\ket{v}_{R}\ket{\mathcal{D}_x(v)}_{D}.$$
This mapping is equivalently stated as:
\begin{eqnarray*}
\mathcal{U}_D:\ket{u}_{L}\ket{v}_{R}\sum_{s\in\Omega'}\sqrt{\pi'_s}\ket{s}_{L'}\ket{\mathcal{D}_x^u(s)}_{D'} &\mapsto & 
\ket{u}_{L}\ket{v}_{R}\sum_{s\in\Omega'}\sqrt{\pi'_s}\ket{s}_{L'}\ket{\mathcal{D}_x^v(s)}_{D'},
\end{eqnarray*}
which can be implemented by an operation with the action:
$$\ket{u}_{L}\ket{v}_{R}\ket{\mathcal{D}_x^u(s)}_{D'}\mapsto \ket{u}_{L}\ket{v}_{R}\ket{\mathcal{D}_x^v(s)}_{D'}.$$ 
The most important point to emphasize about this outer update is that we must update from the inner initial state with respect to some outer state $u$ (contained in $\ket{\mathcal{D}_x(u)}$), to the inner initial state with respect to some other outer state $v$ (contained in $\ket{\mathcal{D}_x(v)}$).

\paragraph{Checking} To check whether a state $u$ is in $M_x$, we want to implement the reflection
$$\mathcal{C}:\ket{u}\ket{\mathcal{D}_x(u)}\mapsto \left\{\begin{array}{ll}
-\ket{u}\ket{\mathcal{D}_x(u)} & \mbox{if }u\in M_x\\
\ket{u}\ket{\mathcal{D}_x(u)} & \mbox{else}\end{array}\right.,$$
which, in this case, is:
\begin{eqnarray*}
\mathcal{C}:\ket{u}\sum_{s\in\Omega'}\sqrt{\pi'_s}\ket{s}\ket{\mathcal{D}_x^u(s)}&\mapsto &
\left\{\begin{array}{ll}
\ket{u}\left(-\sum_{s\in\Omega'}\sqrt{\pi'_s}\ket{s}\ket{\mathcal{D}_x^u(s)}\right) & \mbox{if }M_x^u\neq\emptyset\\
\ket{u}\left(\sum_{s\in\Omega'}\sqrt{\pi'_s}\ket{s}\ket{\mathcal{D}_x^u(s)}\right) & \mbox{else}\end{array}\right. .
\end{eqnarray*}
This is simply a reflection about the initial state of the inner walk with respect to $u$. Importantly, this preserves the inner initial state for future use, so we needn't set it up again. We accomplish this reflection 
by way of \thm{mnrs1} applied to the inner walk. Let $\mathsf{U}'=\mathsf{U}'(P',\mathcal{D}_x^u)$ be an upper bound (over all inputs $x$ and  states $u$) on the update cost of the inner walk. That is, $\mathsf{U}'$ is the cost of implementing
$$\mathcal{U}_P':\ket{s}_{L'}\ket{0}_{R'}\mapsto \ket{s}_{L'}\sum_{t\in\Omega'}\sqrt{P'_{st}}\ket{t}_{R'}
\quad\mbox{and}\quad
\mathcal{U}_D':\ket{s}_{L'}\ket{t}_{R'}\ket{\mathcal{D}_x^u(s)}_{D'}\mapsto \ket{s}_{L'}\ket{t}_{R'}\ket{\mathcal{D}_x^u(t)}_{D'}.$$

Let $\mathsf{C}'=\mathsf{C}'(M_x^u)$ be the checking cost of the inner walk, that is, the cost of implementing the reflection
$$\ket{s}\ket{\mathcal{D}_x^u(s)}\mapsto\left\{\begin{array}{ll}
-\ket{s}\ket{\mathcal{D}_x^u(s)} & \mbox{if }s\in M_x^u\\
\ket{s}\ket{\mathcal{D}_x^u(s)} & \mbox{else}\end{array}\right. .$$

Let $\delta'$ denote the spectral gap of $P'$ and $\eps'$ denote a lower bound on $\pi'(M_x^u)$ for non-empty $M_x^u$.

\begin{lemma}
Let $\widehat{\mathsf{C}}$ be the cost of implementing $\mathcal{C}$ with bounded error. Then $\widehat{\mathsf{C}}= O\left(\frac{1}{\sqrt{\eps'}}\left(\frac{1}{\sqrt{\delta'}}\mathsf{U}'+\mathsf{C}'\right)\right).$
\end{lemma}
\begin{proof}
This follows directly from \thm{mnrs1}, since $\ket{\mathcal{D}_x(u)}$ is the initial state of the inner walk (with respect to $u$).
\end{proof}

\begin{theorem}\label{thm:twolevel}
We can decide whether $M_x$ is non-empty, with bounded error, in quantum query complexity
$$\tilde{O}\left(\mathsf{S}+\frac{1}{\sqrt{\eps}}\left(\frac{1}{\sqrt{\delta}}\mathsf{U}+\frac{1}{\sqrt{\eps'}}\left(\frac{1}{\sqrt{\delta'}}\mathsf{U}'+\mathsf{C}'\right)\right)\right).$$
\end{theorem}
\begin{proof}
We can reduce the error in $\mathcal{C}$ from bounded error to error that is inverse polynomial at the cost of an additional log factor, as mentioned in \sec{prelim}.
Since we only run for a polynomial number of steps, the result follows.
\end{proof}

This construction can be applied inductively to achieve nesting to any depth $k$, though the number of terms in the complexity will be proportional to $k$. Though the generalization to $k$ levels is fairly straightforward, we give the explicit construction in \sec{gen}.

\subsection{Averaging costs} In this section we prove a lemma about averaging costs of operations that will be useful in some of our nested walk applications. In particular, we sometimes have update costs which are much lower on average than in the worst case, and we would like to incur the average, rather than worst case update cost.

Let $\{U_i\}$ be a set of unitaries acting on the same space. If we know how to exactly implement each $U_i$ with $c_i$ queries, then we can exactly implement the controlled unitary $U = \sum_i |i\>\<i| \otimes U_i$ with $\max_i c_i$ queries. However, if the state we are implementing this unitary on has a nearly uniform distribution over the control register, then intuitively it seems like we should be able to approximately implement the unitary with a cost that goes like the average instead of the maximum. The following lemma makes this precise. 

\begin{lemma}
\label{lem:averagecost}
Let $q_i$ be the query complexity of exactly implementing unitary $U_i$. Let $U$ be the controlled unitary $U = \sum_i |i\>\<i| \otimes U_i$. For any prescribed upper bound $q$, we can convert the state $|s\> = \sum_i \alpha_i  |i\>|s_i\>$ to a state $|t\>$ using at most $q$ queries, such that $|t\>$ has inner product at least $1-\epsilon_q$ with $U|s\> = \sum_i \alpha_i  |i\>U_i|s_i\>$, where $\epsilon_q \leq \sum_{i:q_i > q} |\alpha_i|^2$. In particular, with $O(\sum_i |\alpha_i|^2 q_i)$ queries, we can achieve any constant $\epsilon$.
\end{lemma}
\begin{proof}
The algorithm for implementing this conversion is simple: We only implement those matrices $U_i$ whose cost is less than the prescribed upper bound. Otherwise we implement, for example, the identity matrix. More precisely, we implement the unitary $\tilde{U} = \sum_{i:q_i \leq q} |i\>\<i| \otimes U_i + \sum_{i:q_i > q} |i\>\<i| \otimes I$. Clearly this unitary costs at most $q$ queries to implement.

The inner product between $U|s\>$ and $\tilde{U}|s\>$ is $\sum_{i:q_i \leq q} |\alpha_i|^2 \<s_i|s_i\> + \sum_{i:q_i > q} |\alpha_i|^2 \<s_i|U|s_i\>$, which is at least $\sum_{i:q_i \leq q} |\alpha_i|^2$. Thus $\epsilon_q \leq \sum_{i:q_i > q} |\alpha_i|^2$.

To show the final portion of the theorem, consider the random variable that takes value $q_i$ with probability $|\alpha_i|^2$. By Markov's inequality, the probability that this random variable takes a value greater than $k$ times the mean is less than $1/k$. So, for any $k\geq 1$, if we set $q = k(\sum_i |\alpha_i|^2 q_i)$, we can achieve $\epsilon \leq 1/k$.
\end{proof}

Note that Markov's inequality uses very minimal information about the structure of the random variable. If the random variable has more structure, such as being tightly concentrated around the mean, then we can get much better estimates. 
The lemma is simple and constructive, and thus it can also be made efficient in terms of time and space if the approximate unitary $\tilde{U}$ 
can be implemented efficiently.

\subsection{Two applications to triangle finding}
\label{sec:extensions}

We can now present our first two nested quantum walk algorithms, both for the triangle finding problem.
We first focus on the the learning graph of Belovs \cite{bel11} since it is more challenging to implement due to the sparsification idea.
Later we will sketch the quantum walk algorithm corresponding to the more recent and more efficient learning graph in~\cite{LMS12}.

\begin{theorem}
\label{thm:triangle}
There is a nested quantum walk algorithm for triangle finding with quantum query complexity $\tilde{O}(n^{\frac{35}{27}})$.
\end{theorem}
\begin{proof}
Let the input graph be called $G$. The algorithm consists of two quantum walks, which we will call the inner and outer walk. 

\paragraph{Outer walk} The outer walk is a quantum walk on $J(n,r)$, just as in the previous triangle finding algorithm, where  $1\leq r \leq n$ is a parameter to be chosen later. Each state of the Johnson graph corresponds to an $r$-subset of the vertex set of $G$. 
A state of the Johnson graph, $R$, is said to be marked if it contains two triangle vertices. This fully describes the Markov chain of the outer walk; we now have to describe the data structure associated with each state and the checking subroutine, which is the inner walk. As described in  \sec{twolevel}, the data structure associated with each state of the outer walk will be the initial state of the inner walk.

\paragraph{Inner walk} The inner walk is on $J\left(\binom{r}{2},s\binom{r}{2}\right)$, where $1/r < s \le 1$ is a parameter to be chosen later.  A state, $T$, of the inner walk is an $s\binom{r}{2}$-sized subset of the potential edges of an $r$-vertex graph. 
The data structure associated with a state $T$ with outer state $R$, $\mathcal{D}^R(T)$, is the subgraph $G_R(T)$. A state $T$ is said to be marked if $G_R(T)$ contains a triangle edge. The checking subroutine for this walk is the same as the checking step in the quantum walk described in \sec{framework}: we search for a vertex outside of $R$ that forms a triangle with one of the edges in $G_R(T)$. 

\paragraph{Outer walk data structure} We now return to describing the data structure of the outer walk. The quantum data structure corresponding to any state $R$ of the outer walk will be the initial state of the inner quantum walk, i.e., the uniform superposition over all states of the inner walk. For the inner walk, the initial state is $\displaystyle |\pi^R_2\> = \sum_{T \in \Omega_2} \frac{1}{\sqrt{|\Omega_2|}} |T\>|G_R(T)\>$, where $\Omega_2$ is the state space of the inner walk, and $\displaystyle |\Omega_2| = {\binom{\binom{r}{2}}{s\binom{r}{2}}}$.

\paragraph{Inner walk cost analysis} We can now compute the costs of the various operations in the inner walk. In the inner walk, the checking step is the same as the one described in \sec{warmingup} and thus costs $\mathsf{C}_2 = O(\sqrt{n}r^{2/3})$. The update operation just deletes and adds a new edge, giving $\mathsf{U}_2 = O(1)$. The spectral gap of the Johnson walk is $\delta_2 = \Omega(1/sr^2)$ and the fraction of marked states, assuming the outer state is marked, is $\eps_2 \geq s$. Thus the cost of performing the inner walk, given the inital state for the walk, which is the checking cost for the outer walk, is 
\begin{eqnarray*}
\mathsf{C}_1 &=& \tilde{O}\(\frac{1}{\sqrt{\eps_2}}\(\frac{\mathsf{U}_2}{\sqrt{\delta_2}}+\mathsf{C}_2\)\)\\
& =&\tilde{O}\(\frac{1}{\sqrt{s}}\(\sqrt{sr^2}+\sqrt{n}r^{2/3}\)\)\\
& = &\tilde{O}\(r + \sqrt{\frac{n}{s}}r^{2/3}\).\end{eqnarray*}
 Since the checking step with cost $\mathsf{C}_2$ is only bounded error, we need to amplify the success probability, which leads to the log factor subsumed in the $\tilde{O}$ notation.

\paragraph{Outer walk cost analysis} Now we can move to the outer walk. The superposition  corresponding to a state in the outer walk is 
$$|R\>|\pi^R_2\> = |R\>\sum_{T \in \Omega_2} \frac{1}{\sqrt{|\Omega_2|}} |T\>|G_R(T)\>.$$
 Setting up this superposition  clearly costs $\mathsf{S}_1 = sr^2$ queries. A state is considered marked if it contains both vertices of a triangle edge. Thus $\eps_1 = \Omega\(\(\frac{r}{n}\)^2\)$. The spectral gap of this walk is $\delta_1 = \Omega(1/r)$.

Finally, we have to compute the update cost. The update cost is the cost of performing the map 
$$|R\>\otimes |R'\>\otimes \sum_{T \in \Omega_2} \frac{1}{\sqrt{|\Omega_2|}} |T\>|G_R(T)\>$$
  to  
$$|R\>\otimes |R'\>\otimes \sum_{T \in \Omega_2} \frac{1}{\sqrt{|\Omega_2|}} |T\>|G_{R'}(T)\>.$$
 Since  $R$ and $R'$ are adjacent states in a Johnson graph, $R$ has one vertex not in $R'$, and $R'$ has one vertex not in $R$. Let us call these vertices $a$ and $a'$ respectively. The cost of updating  $|G_{R}(T)\>$ to $|G_{R'}(T)\>$ is therefore $\textrm{deg}(a)+\textrm{deg}(a')$. In the worst case the degree of a vertex is $r$. However, on average any vertex has degree $sr$.  Now we can invoke \lem{averagecost}. For a given query upper bound $q$, $\epsilon_q$ is the probability that a randomly selected vertex will have degree greater than $q$. If we set $q = 7(sr)+100\log n$, then we need to compute the probability that the degree of a vertex is more than $7$ times its mean (the extra term is to deal with the case when $sr$ is very small). The degree of a vertex is a random variable distributed according to the hypergeometric distribution, for which strong tail bounds are known. By Theorem 2.10 and eq. (2.11) of~\cite{JLR11}, $\epsilon_q \leq \exp(-q)$ for any $q$ at least $7$ times the mean. In our case, this gives $\epsilon_q\leq \exp(-7sr)\frac{1}{n^{100}}$, which is sufficient for our purposes.

Thus the final query complexity is
\begin{eqnarray*}
& & \tilde{O}\(\mathsf{S}_1 + \frac{1}{\sqrt{\eps_1}}\(\frac{\mathsf{U}_1}{\sqrt{\delta_1}}+\mathsf{C}_1\)\) \\
& = & \tilde{O}\(sr^2 + \frac{n}{r}\(sr\sqrt{r}+r + \sqrt{\frac{n}{s}}r^{2/3}\)\) \\
& = & \tilde{O}\(sr^2 + ns\sqrt{r} + n + \frac{n^{3/2}}{\sqrt{s}r^{1/3}}\),
\end{eqnarray*}
which is optimized by choosing $r=n^{2/3}$ and $s=n^{-1/27}$, giving a final query complexity of $\tilde{O}(n^{35/27})$.
\end{proof}

We note that the optimal values of $r$ and $s$ are the same as those obtained in the learning graph upper bound. In fact, if we add some unnecessary layers of nesting by breaking up the outer walk into several nested walks, we can obtained the learning graph complexity term for term. This suggests a natural correspondence between these two frameworks.

To emphasize a key ingredient in our framework, we note that the inner update is much cheaper than the outer update, since the inner update involves changing a single edge, whereas the outer update involves changing all edges adjacent to a single vertex. This is ideal, since the inner update operation is performed much more frequently than the outer update.

\bigskip
A very recent unpublished result of Lee, Magniez, and Santha \cite{LMS12} constructs a learning graph for triangle finding with query complexity $O(n^{9/7})=O(n^{1.286})$. We can easily convert their learning graph to a very simple nested walk. Due to its simplicity, it is perhaps less illustrative of the power of the framework than the original learning graph for triangle finding, however, it illustrates the cost savings that can be obtained by performing expensive updates less frequently (by keeping them in the outer walk).

\begin{theorem}\label{thm:triangle2}
There is a nested quantum walk algorithm for triangle finding with query complexity $\tilde{O}(n^{9/7})$.
\end{theorem}

\begin{proof}
We will make use of the fact that collision finding on an unbalanced bipartite graph on classes of size $r_1$ and $r_2$ with $r_1\leq r_2\leq r_1^2$ has quantum query complexity $O((r_1r_2)^{1/3})$ (see \app{gc}).

\paragraph{Outer walk}The outer walk is on $J(n,r_1)$, for $1\leq r_1\leq n$. A state $R_1$ is marked if it contains a triangle vertex. Thus, we have $\delta_1 =\Omega( \frac{1}{r_1})$ and $\eps_1 = \Omega(\frac{r_1}{n})$. 

\paragraph{Inner walk}The inner walk is on $J(n,r_2)$, for $1\leq r_2\leq n$. Suppose the outer state, $R_1$, is marked, so it contains the vertex from some triangle. A state $R_2$ is marked if it contains a vertex from the same triangle. 

The data structure associated with a state $R_2$ with outer state $R_1$, $\mathcal{D}^{R_1}(R_2)$, is the subgraph of $G$ containing exactly those edges with one endpoint in $R_1$, and one endpoint in $R_2$, $G_{R_1, R_2}$. Thus, $\delta_1=\Omega(\frac{1}{r_2})$ and $\eps_2 = \Omega(\frac{r_2}{n})$. To update, we must query $r_1$ potential edges from a new vertex to $R_1$, so $\mathsf{U}_2 = r_1$. To check if a state is marked, we search for a third triangle vertex, using a graph collision subroutine, with total cost $\mathsf{C}_2 = \sqrt{n}(r_1r_2)^{1/3}$. 

\paragraph{Outer walk data structure} The quantum data structure corresponding to a state $R_1$ of the outer walk is $\sum_{R_2}\frac{1}{\sqrt{\abs{\Omega_2}}}\ket{R_2}\ket{G_{R_1, R_2}}$. The cost of learning $G_{R_1,R_2}$ is $\mathsf{S}_1=r_1r_2$, and the cost of updating $G_{R_1,R_2}$ to $G_{R_1',R_2}$ is $\mathsf{U}_1 = r_2$, since we must query $r_2$ potential edges between the new vertex and all of $R_2$.

The final query complexity is (neglecting log factors)
\begin{eqnarray*}
&&\mathsf{S}_1+\frac{1}{\sqrt{\eps_1}}\left(\frac{1}{\sqrt{\delta_1}}\mathsf{U}_1+\frac{1}{\sqrt{\eps_2}}\left(\frac{1}{\sqrt{\delta_2}}\mathsf{U}_2+\mathsf{C}_2\right)\right)\\
&=&r_1r_2+\sqrt{\frac{n}{r_1}}\left(\sqrt{r_1}r_2+\sqrt{\frac{n}{r_2}}\left(\sqrt{r_2}r_1+\sqrt{n}(r_1r_2)^{1/3}\right)\right)\\
&=& r_1r_2 + \sqrt{n}r_2 + n\sqrt{r_1} + \frac{n^{3/2}}{r_1^{1/6}r_2^{1/6}}.
\end{eqnarray*}
Setting $r_1 = n^{4/7}$ and $r_2 = n^{5/7}$ gives the claimed upper bound. 
Note that the optimal values have $r_2$ larger than $r_1$, making the inner update cheaper than the outer update. The inner update is run more often, not only because $\frac{1}{\sqrt{\delta_2}}>\frac{1}{\sqrt{\delta_1}}$, but more fundamentally, because the term $\frac{1}{\sqrt{\delta_2}}\mathsf{U}_2$ has a factor of $\frac{1}{\sqrt{\eps_1\eps_2}}$, whereas the term $\frac{1}{\sqrt{\delta_1}}\mathsf{U}_1$ only has a factor of $\frac{1}{\sqrt{\eps_1}}$.
\end{proof}

For this walk, in some sense, the inner marked set is independent of the outer state. To see this, consider the case where there is a single triangle, and note that this is the edge case from which we calculate a lower bound on the probability that a state is marked. In the case of a single triangle, we can imagine the inner and outer walks being on independent sets of vertices, $R_1$ and $R_2$, each of which is marked if it contains a triangle vertex. (Of course, they should not contain the same triangle vertex.) In some sense then, this algorithm does not use the full power of the model.

\section{Multi-level nested quantum walks}
\subsection{Generalization of the framework}
\label{sec:gen}

The simple two-level nesting from \sec{framework} can be generalized in a straightforward way to $k$ levels of nesting.

Consider $k$ walks $\{(P^{(i)},\Omega^{(i)},\pi^{(i)})\}_{i=1}^k$. Suppose the \th{i} walk has marked sets $$\{M_x^u\}_{u\in\Omega^{(1)}\times\dots\times\Omega^{(i-1)}}$$
(with $M_x$ implicitly being the unique marked set for the outermost walk) and data mappings
$$\{\mathcal{D}_x^u:\Omega^{(i)}\rightarrow \mathcal{H}_{D^{(i)}}\}_{u\in\Omega^{(1)}\times\dots\times\Omega^{(i-1)}}.$$
Suppose that for all $i=1,\dots,k$ and all $u\in\Omega^{(1)}\times\dots\times\Omega^{(i-1)}$, we have
$$M_x^u = \{v\in\Omega^{(i)}:M_x^{(u,v)}\neq\emptyset\},$$
and for $i=1,\dots,k-1$, $u\in\Omega^{(1)}\times\dots\times\Omega^{(i-1)}$ and $v\in\Omega^{(i)}$, we have $\mathcal{H}_{D^{(i)}}=\mathcal{H}_{\Omega^{(i+1)}}\otimes\mathcal{H}_{D^{(i+1)}}$, and
$$\ket{\mathcal{D}_x^u(v)} = \sum_{w\in\Omega^{(i+1)}}\sqrt{\pi^{(i+1)}_w}\ket{w}\ket{\mathcal{D}_x^{(u,v)}(w)}.$$

Let $\delta_i$ denote the spectral gap of $P^{(i)}$, and $\eps_i$ be a lower bound for non-empty $\pi^{(i)}(M_x^u)$ for all $u\in\Omega^{(1)}\times\dots\times\Omega^{(i-1)}$. Let $\mathsf{S}$ denote the cost of constructing 
$$\sum_{u\in\Omega^{(1)}}\sqrt{\pi^{(1)}_u}\ket{u}\ket{\mathcal{D}_x(u)}.$$
Let $\mathsf{U}_i$ be the update cost of the \th{i} walk. Let $\mathsf{C}$ be the cost of implementing, for $w\in\Omega^{(k)}$, and $u\in\Omega^{(1)}\times\dots\times\Omega^{(k-1)}$,
$$\ket{w}\ket{\mathcal{D}_x^u(w)}\mapsto\left\{\begin{array}{ll}
-\ket{w}\ket{\mathcal{D}_x^u(w)} & \mbox{if } w\in M_x^u\\
\ket{w}\ket{\mathcal{D}_x^u(w)} & \mbox{else}\end{array}\right. .$$

\begin{corollary}
We can decide whether $M_x$ is non-empty with bounded error in 
\begin{eqnarray}&&\tilde{O}\(\mathsf{S}+\sum_{i=1}^k\left(\prod_{j=1}^i\frac{1}{\sqrt{\eps_j}}\right)\frac{1}{\sqrt{\delta_i}}\mathsf{U}_i+\left(\prod_{j=1}^k\frac{1}{\sqrt{\eps_j}}\right)\mathsf{C}\)\label{eq:gen}\end{eqnarray}
queries.
\end{corollary}

This is not the most general nested walk. Rather than have a single \th{i} walk $(P^{(i)},\Omega^{(i)},\pi^{(i)})$, we could have a family of \th{i} walks, $\{(P^u,\Omega^u,\pi^u)\}_{u\in\Omega^{(1)}\times \dots\times\Omega^{(i-1)}}$, where $\Omega^{(i)}=\bigcup_{u\in\Omega^{(1)}\times\dots\times\Omega^{(i-1)}}\Omega^u$. However, such a general framework sacrifices the simplicity that is a feature of nested quantum walks. 

\subsection{Application to subgraph containment}\label{sec:subgraph}
It is simple to extend our quantum walk algorithms for triangle finding to obtain the more general learning graph upper bounds for subgraph containment~\cite{LMS11,zhu11,LMS12}. 
We now give an explicit construction of nested quantum walks corresponding to the learning graphs of~\cite{LMS12}, matching their complexity up to polylogarithmic factors, and even potentially improving their result for some cases.
For simplicity consider  the case of undirected graphs with no self-loops.
Let $H$ be a graph on vertices $\{1,2,\ldots,k\}$. Then a graph $G$ contains a copy of $H$ if $G$ induces 
a graph that is isomorphic to $H$ on some subset of $k$ vertices.

\paragraph{The learning graph construction} In~\cite{LMS12}, a meta-language is introduced for designing a learning graph that detects a copy of $H$ in $G$.
Instructions are \texttt{Setup}, \texttt{LoadVertex($i$)} for $i\in [k]$ and \texttt{LoadEdge($i,j$)} for $\{i,j\}\in E(H)$.
A `program' is an ordering of $\{{\tt Setup}\}\cup\{\mbox{{\tt LoadVertex($i$)}}\}_{i\in [k]}\cup\{\mbox{{\tt LoadEdge($i,j$)}}\}_{\{i,j\}\in E(H)}$ such that {\tt Setup} appears first and for each edge $\{i,j\}$ of $H$, instruction \texttt{LoadEdge($i,j$)} occurs in the program
after instructions \texttt{LoadVertex($i$)} and \texttt{LoadVertex($j$)}.

Parameters are vertex-set sizes $r_i\geq 1$ for each vertex $i$ of $H$, and average degrees $d_{ij}\geq 1$ for each edge $\{i,j\}$ of $H$. These parameters are not related to $H$ but parameterize the resulting learning graph.
Assuming some natural conditions on these parameters, one can derive an algorithm whose complexity only depends on them. Optimizing them can be done via a linear program (see~\cite{LMS12} for a link to the program).
We assume that $r_i\leq n$,  $d_{ij} \le \max (r_i, r_j)$, and that: 
\begin{equation} \label{eq:lg-hyp} 
\text{\em For all $i$ there exists $j$ such that $\{i,j\} \in E(H)$ and $d_{ij} (2r_j+1)/(2r_i+1) \ge 1$.}
\end{equation}

Each instruction has a cost that propagates to later instructions, and its total effective cost. 
The first one is called {\em global cost}
and the second one {\em local cost}. In \cite{LMS12}, a table is given, nearly identical to Table \ref{tab:cost}, providing global and local costs for each of the three types of instructions. Given a program of $\tau=k+\abs{E(H)}+1$ instructions, if $l_t$ denotes the local cost of instruction $t$, and $g_t$ the global cost, then the resulting upper bound for deciding if $H$ is in $G$ is
\begin{eqnarray} 
&&\sum_{t=1}^{\tau}\(\prod_{t'< t}g_{t'}\)l_t.\label{eq:lg}
\end{eqnarray}

\paragraph{The corresponding nested walk} Given a program, we can define a sequence of nested walks as follows. The {\tt Setup} instruction corresponds to our quantum walk setup. Each non-setup instruction corresponds to a quantum walk, with the walk corresponding to the \th{t} instruction being nested at depth $t$; that is, the \th{t} walk is the checking procedure for the \th{(t-1)} walk. 
Each non-setup instruction corresponds to an edge or vertex in $H$, and so each of our walks will as well. 

All walks will be on Johnson graphs. The state of the walk corresponding to a vertex $i$ is described by a subset of vertices $R_i$ of size $r_i$,
and the state of the walk corresponding to  an edge $\{i,j\}$ is described by a subset $T_{i,j}$ of $d_{ij}\min(r_i,r_j)$ potential edges between $R_i$ and $R_j$.
For each edge $\{i,j\}$ of $H$, our database maintains the edges $G(T_{ij})$ of $G$ in $T_{i,j}$.

A state $R_i$ is marked if there exists some copy of $H$ in $G$ such that: for each $R_j$ in a walk at depth \emph{less than} the depth of $R_i$, $R_j$ contains the \th{j} vertex in the subgraph, and for each $T_{uv}$ at depth less than the depth of $R_i$, $T_{uv}$ contains edge $\{u,v\}$ in the subgraph, and finally, $R_i$ contains the \th{i} vertex in the subgraph. A state $T_{ij}$ has a similar condition for being marked, except $T_{ij}$ must contain edge $\{i,j\}$ in the subgraph.

\paragraph{Setup}
To setup, we must create
the uniform superposition $\ket{\pi}$ of subgraphs $G(T_{ij})$, one for each edge $\{i,j\}\in E(H)$.
The total cost is
$$
\mathsf{S}=\sum_{\{i,j\} \in H} \min (r_i, r_j) d_{ij}.
$$

\paragraph{Vertex walks}
To update a walk on vertex set $R_i$ we remove a vertex and add a new one. To update the data, we must update the graphs $G(T_{ij})$, for each $j$ adjacent to $i$ in $H$. 
By \lem{averagecost}, this complexity is related to the average degree in $R_i$, which we can calculate as $d_{ij}\min(1,r_j/r_i)$.  The total cost for this type of update is then 
$$
\mathsf{U}=\sum_{j: \{i,j\}\in H} d_{ij} \min(1,{r_j}/{r_i}). \Big.
$$
Observe that this quantity is always $\Omega(1)$, by the imposed condition in~\eqref{eq:lg-hyp}. 

We can easily calculate that the proportion of marked states in this type of walk is $\eps = r_i/n$, and the spectral gap is $\delta=1/r_i$.

\paragraph{Edge walks}
Lastly, to update a walk on edge set $T_{ij}$ we simply replace one edge in $T_{ij}$ and update $G(T_{ij})$.
Thus the cost of this type of update is simply $\mathsf{U}=1$. For this type of walk we have $\eps=d_{ij}/\max(r_i,r_j)$ (since $\eps$ is the probability that a particular edge of $R_i\times R_j$ is in $T_{ij}$) and $\delta = 1/\min(r_i,r_j)$.

\paragraph{Checking} A state of the deepest walk is marked if the data encodes an entire copy of $H$ in $G$. Since we can detect this from the data with no further queries to $G$, the final checking cost is $\mathsf{C}=0$. 

\paragraph{Global and local costs} We can rephrase our costs for each walk into global and local costs. The global cost is simply the part of the cost that propagates forward to deeper walk costs. For the setup (which we can consider as a walk on a complete graph in which every state is marked), this is simply $1$, whereas for other walks, the global cost is $\frac{1}{\sqrt{\eps}}$ (see \eq{gen}). Similarly, we can consider local costs. For the setup this is defined as $\mathsf{S}$, whereas for each other walk, it is defined as $\frac{1}{\sqrt{\eps\delta}}\mathsf{U}$.  
We list all global and local costs in \tab{cost}.
\begin{table}

\begin{center}
\begin{tabular}{|r|c|c|}
\hline
Instruction & Global Cost & Local Cost \\
\hline\hline
\texttt{Setup} & 1 & $\mathsf{S}=\displaystyle\sum_{\{i,j\} \in H} \min\{r_i, r_j\} d_{ij}$ \\
\hline
\texttt{LoadVertex($i$)} & $\frac{1}{\sqrt{\eps}}=\sqrt{n/r_i}$  & $\frac{1}{\sqrt{\eps\delta}}\mathsf{U}=\displaystyle\sqrt{n} \times \!\!
\sum_{j: \{i,j\}\in H} d_{ij} \min(1,{r_j}/{r_i}) $ \\
\hline
\texttt{LoadEdge($i,j$)} & $\frac{1}{\sqrt{\eps}}=\sqrt{\max\{r_i, r_j\} / d_{ij}}$ & $\frac{1}{\sqrt{\eps\delta}}\mathsf{U}=\sqrt{r_i  r_j}$ \\
\hline
\end{tabular}
\end{center}
\caption{Global and local costs for quantum walks corresponding to each type of instruction.}\label{tab:cost}

\end{table}

Comparing \eq{gen} and \eq{lg}, we see that the upper bounds we achieve for subgraph detection are the same as those of \cite{LMS12} (up to polylogarithmic factors), as a function of the global and local costs. Furthermore, our expressions for global and local costs are almost identical to those of \cite{LMS12}. The only exception is that, depending on some sparsity condition, namely whether $d_{ij} \min(r_i,r_j)<\max(r_i,r_j)$, the local cost of \texttt{LoadEdge($i,j$)}
has two different expressions in~\cite{LMS12}: either $\sqrt{r_i  r_j}$ (when the condition holds) or $\max(r_i, r_j)$
(otherwise). Therefore, ignoring polylogarithmic factors, our total complexity is the same as that of \cite{LMS12} or potentially better in some cases, because of our potentially better edge-local cost, though we know of no particular graph for which we have asymptotically better complexity.
We remark that we believe that this improvement in edge-local cost could also be made in the learning graph construction of \cite{LMS12}, however, it isn't immediately clear how to do so in that setting, whereas in the quantum walk setting it is natural and immediate.

\section{Notes on other learning graphs}

In~\cite{bel12}, a learning graph algorithm for graph collision is presented. We note that this is, in fact, based on an unpublished quantum walk algorithm due to Ambainis~\cite{bel12}, which has the same complexity, up to log factors.

As to the possibility of reproducing the recent learning graph upper bound for $k$-distinctness~\cite{bel12}, while this upper bound is achieved via a method that deviates significantly from the original learning graph framework, we do believe that we should be able to construct a quantum walk algorithm with matching complexity.

\section*{Acknowledgments}

The authors would like to thank Andrew Childs, Troy Lee, Ashwin Nayak and Miklos Santha for helpful discussions.

\bibliographystyle{alpha}
\bibliography{refs}

\appendix

\section{Graph collision subroutine}\label{app:gc}

In this section, we show how to compute graph collision on a bipartite graph with classes of size $r_1$ and $r_2$ in quantum query complexity $O\((r_1r_2)^{1/3}\)$. The algorithm is a straightforward generalization of the standard graph collision algorithm of \cite{MSS07}, based on the element distinctness algorithm of \cite{amb04}. Note that when $r_1=r_2$, the well-known $O\(r_1^{2/3}\)$ upper bound is achieved.

\begin{lemma}
The quantum query complexity of graph collision on a bipartite graph with classes of size $r_1$ and $r_2$, when $r_1\leq r_2\leq r_1^2$, is in $O\((r_1r_2)^{1/3}\)$.
\end{lemma}
\begin{proof}
Let $G$ be a bipartite graph on classes $R_1$ and $R_2$ with $\abs{R_1}=r_1$ and $\abs{R_2}=r_2$. We proceed as in the standard graph collision algorithm: let $m\leq r_1$ be a parameter to be optimized later. We walk on the product graph $J(r_1,m)\times J(r_2,m)$, where a state $S_1,S_2$ corresponds to a pair of subsets; $S_1$ is an $m$-subset of $R_1$ and $S_2$ is an $m$-subset of $R_2$. Two states $S_1,S_2$ and $S_1',S_2'$ are adjacent in $J(r_1,m)\times J(r_2,m)$ if $S_1$ is adjacent to $S_1'$ in $J(r_1,m)$ and $S_2$ is adjacent to $S_2'$ in $J(r_2,m)$. 

The data associated with a state $S_1,S_2$ is the set of query values for vertices $i\in S_1\cup S_2$. Thus, the setup cost is $\mathsf{S} = O(m)$, and the update cost is $\mathsf{U}=O(1)$. A state is marked if there is a pair $i_1\in S_1$ and $i_2\in S_2$ such that $i_1$ is adjacent to $i_2$ and both are marked. Since the data contains query values for both $i$ and $i'$, checking costs $\mathsf{C}=0$ queries. The proportion of marked states is at least $\eps = \Omega\(\frac{m^2}{r_1r_2}\)$. The spectral gap of $J(r_1,m)\times J(r_2,m)$ is $\delta=\Omega\(\frac{1}{m}\)$. Thus, the query complexity is 
$$O\left(m+\frac{\sqrt{r_1r_2}}{m}\sqrt{m}\right),$$
which is optimized by setting $m=(r_1r_2)^{1/3}$ (which is smaller than $r_1$ when $r_2\leq r_1^2$), yielding quantum query complexity $O\((r_1r_2)^{1/3}\)$.
\end{proof}

\end{document}